\documentclass[11pt]{llncs}

\usepackage{latexsym}
\usepackage{graphicx}
\usepackage{ascmac}
\usepackage{xcolor}
\usepackage{url}
\usepackage{xspace}
\usepackage{amsmath}
\usepackage{fullpage}
\usepackage{algorithm,algorithmic}
\usepackage{comment}

\begin{document}

\definecolor{value}{HTML}{776e65}
\definecolor{twocolor}{HTML}{eee4da}
\definecolor{fourcolor}{HTML}{ede0c8}
\definecolor{val2048}{HTML}{F9F6F2}
\definecolor{bg8}{HTML}{F2b179}
\definecolor{bg16}{HTML}{F59563}
\definecolor{bg32}{HTML}{F67C5F}
\definecolor{bg64}{HTML}{F65e3b}
\definecolor{bg128}{HTML}{EDCF72}
\definecolor{bg256}{HTML}{edcc61}
\definecolor{bg512}{HTML}{edc850}
\definecolor{bg1024}{HTML}{EDC53F}
\definecolor{bg2048}{HTML}{edc22e}
\definecolor{bgboard}{HTML}{bbada0}
\newcommand*{\boxeddigit}[1]{\fbox{\sf #1}\hspace{1mm}}
\newcommand{\two}{\fcolorbox{bgboard}{twocolor}{\sf \textcolor{value}{\fontseries{bx}2}}\hspace{1mm}}
\newcommand{\four}{\fcolorbox{bgboard}{fourcolor}{\sf \textcolor{value}{\fontseries{bx}4}}\hspace{1mm}}
\newcommand{\tileVIII}{\fcolorbox{bgboard}{bg8}{\sf \textcolor{val2048}{\fontseries{bx}8}}\hspace{1mm}}
\newcommand{\tileXVI}{\fcolorbox{bgboard}{bg16}{\sf \textcolor{val2048}{\fontseries{bx}16}}\hspace{1mm}}
\newcommand{\tileXXXII}{\fcolorbox{bgboard}{bg32}{\sf \textcolor{val2048}{\fontseries{bx}32}}\hspace{1mm}}
\newcommand{\tileLXIV}{\fcolorbox{bgboard}{bg64}{\sf \textcolor{val2048}{\fontseries{bx}64}}\hspace{1mm}}
\newcommand{\tileCXXVIII}{\fcolorbox{bgboard}{bg128}{\sf \textcolor{val2048}{\fontseries{bx}128}}\hspace{1mm}}
\newcommand{\tileCCLVI}{\fcolorbox{bgboard}{bg256}{\sf \textcolor{val2048}{\fontseries{bx}256}}\hspace{1mm}}
\newcommand{\tileDXII}{\fcolorbox{bgboard}{bg512}{\sf \textcolor{val2048}{\fontseries{bx}512}}\hspace{1mm}}
\newcommand{\tileMXXIV}{\fcolorbox{bgboard}{bg1024}{\sf \textcolor{val2048}{\fontseries{bx}1024}}\hspace{1mm}}
\newcommand{\tileMMXLVIII}{\fcolorbox{bgboard}{bg2048}{\sf \textcolor{val2048}{\fontseries{bx}2048}}\hspace{1mm}}
\newcommand{\Xtile}{\fcolorbox{bgboard}{bg2048}{\sf \textcolor{val2048}{\fontseries{bx}\X}}\hspace{1mm}}

\newcommand{\onethrees}{\fcolorbox{bgboard}{cyan}{\sf \textcolor{value}{\fontseries{bx}1}}\hspace{1mm}}
\newcommand{\twothrees}{\fcolorbox{bgboard}{red}{\sf \textcolor{value}{\fontseries{bx}2}}\hspace{1mm}}
\newcommand{\threethrees}{\fcolorbox{bgboard}{twocolor}{\sf \textcolor{value}{\fontseries{bx}3}}\hspace{1mm}}
\newcommand{\sixthrees}{\fcolorbox{bgboard}{twocolor}{\sf \textcolor{value}{\fontseries{bx}6}}\hspace{1mm}}
\newcommand{\twelvethrees}{\fcolorbox{bgboard}{twocolor}{\sf \textcolor{value}{\fontseries{bx}12}}\hspace{1mm}}

\newtheorem{observation}[theorem]{Observation}
\newtheorem{fact}[theorem]{Fact}
%\def\proof{\par\noindent{\bf Proof.~}}
%\spn@wtheorem{fact}{Fact}{\itshape}{\rmfamily}

\def\X{T}
\def\NP{{\sf NP}}
\def\FPT{{\sf FPT}}

\title{Threes!, Fives, 1024!, and 2048 are Hard}
\author{%
   Stefan Langerman\thanks{Directeur de Recherche du F.R.S.-FNRS}\inst{1} \and %
   Yushi Uno\inst{2}
}

\institute{
    D\'{e}partement d'informatique, 
    Universit\'{e} Libre de Bruxelles, 
    ULB CP 212, avenue F.D. Roosevelt 50, 1050 Bruxelles, Belgium. 
    \email{stefan.langerman@ulb.ac.be}
\and
    Department of Mathematics and Information Sciences, 
    Graduate School of Science, Osaka Prefecture University, 
    1-1 Gakuen-cho, Naka-ku, Sakai 599-8531, Japan. 
    \email{uno@mi.s.osakafu-u.ac.jp}
}

\maketitle

\begin{abstract}
We analyze the computational complexity of the popular computer games
{Threes!}, {1024!}, {2048} and many of their variants. 
For most known versions expanded to an $m\times n$ board, we show that
it is \NP-hard to decide whether a given starting position can be
played to reach a specific (constant) tile value. 
\end{abstract}

\section{Introduction}

\begin{figure}[hbt]
\centering
\begin{minipage}{.30\linewidth}
\centering
\scalebox{0.16}{\includegraphics{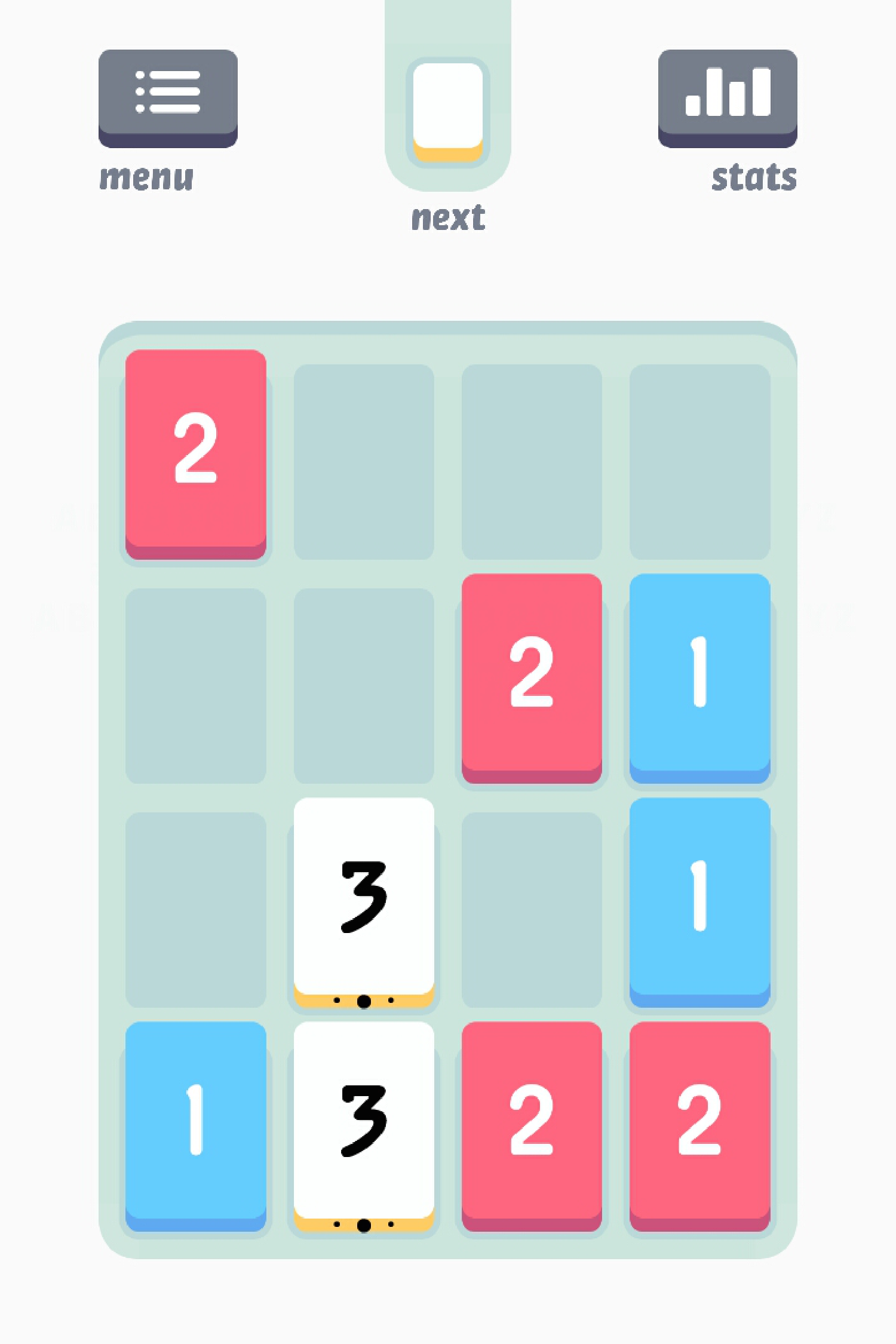}}
\caption{Threes!}
\label{threes}
\end{minipage}
\begin{minipage}{.03\linewidth}
\mbox{}
\end{minipage}
\begin{minipage}{.30\linewidth}
\centering
\scalebox{0.39}{\includegraphics{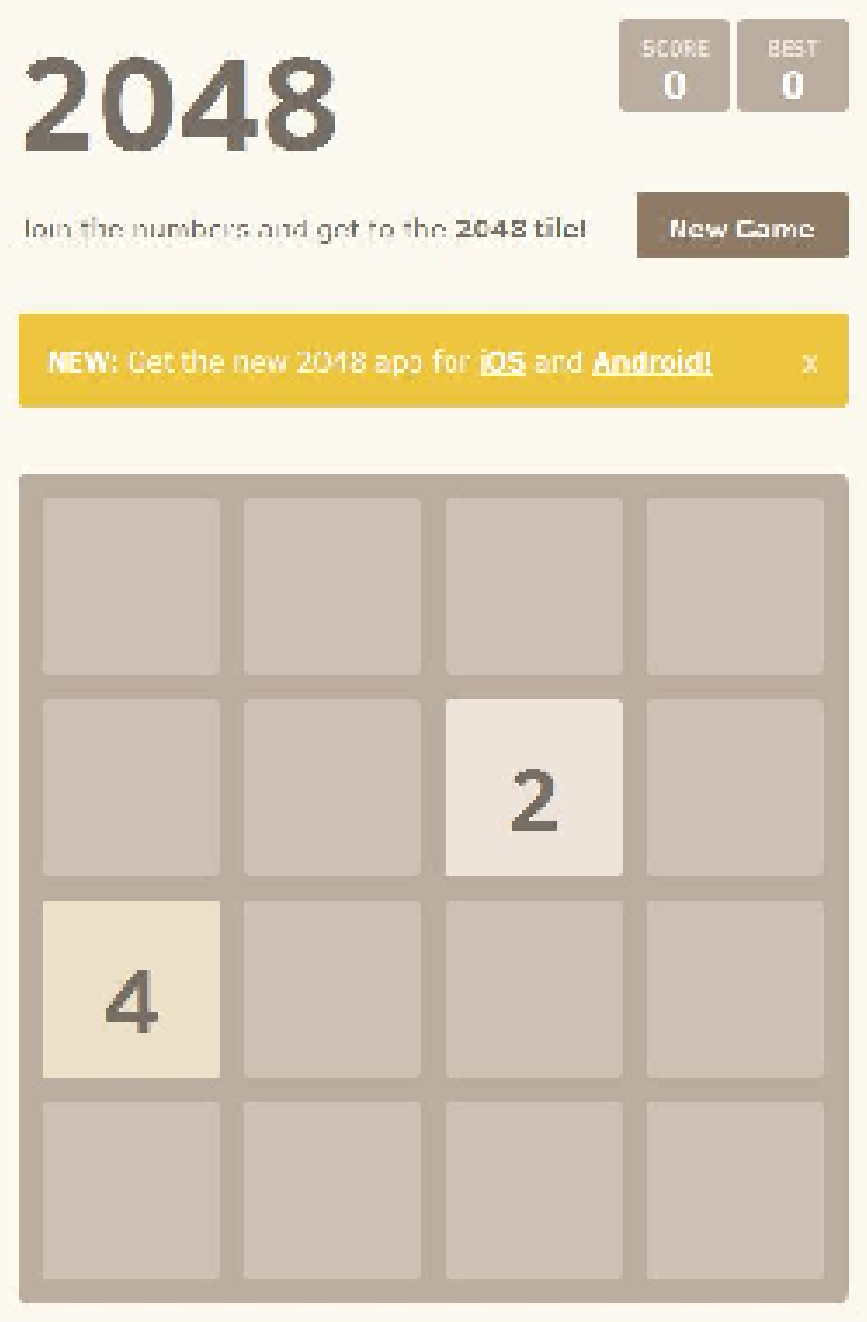}}
\caption{The game 2048: a board and one of its initial configuration.}
\label{2048_original_initial}
\end{minipage}
\begin{minipage}{.03\linewidth}
\mbox{}
\end{minipage}
\begin{minipage}{.30\linewidth}
\centering
\scalebox{0.40}{\includegraphics{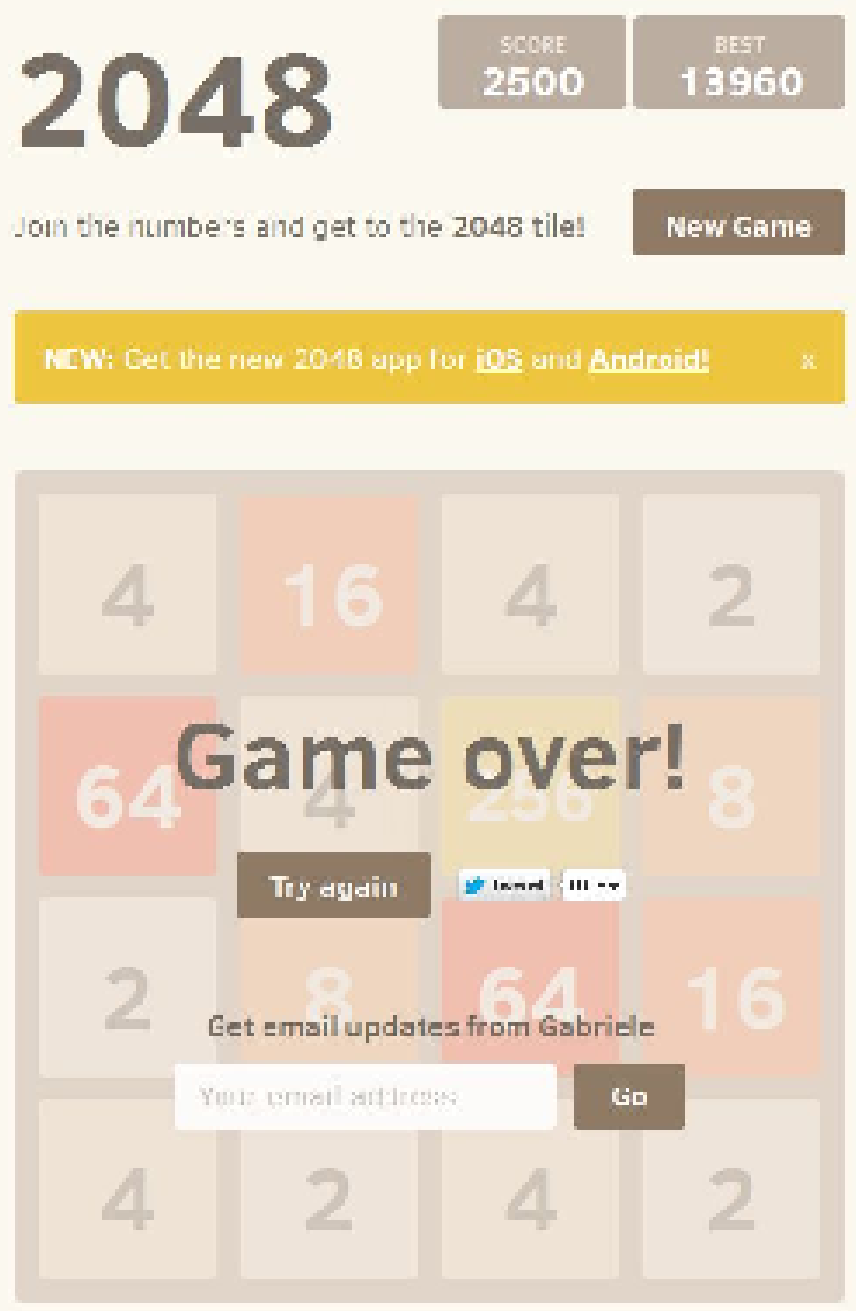}}
\caption{A forbidden (game over) configuration.}
\label{2048_gameover}
\end{minipage}
\end{figure}

%Threes! Jan 23, 2014
%Fives Feb 18, 2014
%1024! Feb 27, 2014
%2048 March 5, 2014
%History: http://asherv.com/threes/threemails/
%https://developer.apple.com/design/awards/2014/Threes/

\emph{Threes!}~\cite{Threes} is a popular puzzle game created by Asher Vollmer, Greg
Wohlwend, and Jimmy Hinson (music), and released by Sirvo for iOS on
January 23, 2014. The game received considerable attention from players,
game critics and game designers. Only a few weeks after its release,
an Android clone \emph{Fives} appeared, and then an iOS clone
\emph{1024!} with slightly modified rules. Shortly after, two open
source web game versions, both called \emph{2048} were released on
github on the same day, one by Saming~\cite{2048saming}, the other by
Gabriele Cirulli~\cite{2048cirulli}. Since then over a hundred new variant have
been catalogued~\cite{2048variants}. In December 2014, Threes! received the Apple Game of the
Year and the Apple Design award.

\begin{description}
\item[2048] 
We first describe Cirulli's 2048~\cite{2048cirulli} (or just 2048 for short) which has a slightly simpler set of
rules (see Fig.~\ref{2048_original_initial}). 
The game is played on a $4\times 4$ square grid {\it board},
consisting of 16 {\it cells}. 
During the game, each cell is either empty or contains a \emph{tile}
bearing a \emph{value} which is a power of two.
When the game begins, a (random) starting \emph{configuration} of
tiles is placed on the board. 
Then, in every turn, 
one plays a move by indicating one of four directions, 
{\rm up}, {\rm down}, {\rm left} or {\rm right}, 
and then each numbered tile moves in that direction, either to the
boundary of the board (a \emph{wall}) or until it hits another tile. 
When two tiles of value $K$ hit, they merge to become a single tile of
value $2K$. 
If three or more tiles with the same value hit, 
they merge two by two, starting with the two closest to the wall in
the direction of the move. 
If no tile can move in some direction (e.g., all tiles touch the
wall), then that move is \emph{invalid}.
After each turn, a new tile of value \two or \four 
appears in a random empty cell.
The objective of the game is to make a tile of value 2048, and/or to
maximize the \emph{score} defined as the sum of all new tiles created
by merges during the game. 
If during the game, the board is completely filled and
no move is valid, then
the game is over and the player loses. 
We call such a configuration {\it forbidden} (Fig.~\ref{2048_gameover}). 

\begin{figure}[hbt]
\centering
\begin{minipage}{.30\linewidth}
\centering
\scalebox{0.40}{\includegraphics{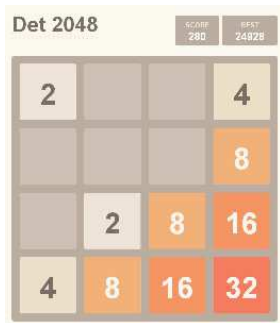}}
\caption{Deterministic 2048.}
\label{2048det}
\end{minipage}
\end{figure}

\item[Threes!] The tiles in Threes! have values $1$, $2$, and 
   $3 \cdot 2^i, i\geq 0$. The tiles $1$ and $2$ combine to form tile $3$, and
   tiles of value $K\geq 3$ combine to form tile $2K$. Another
   important difference is that when performing a move, all tiles move
   in the corresponding direction by at most one cell instead of moving as far as
   possible. For example during a {\rm left} move, a tile next to the left wall (if any) doesn't move (we say
   it is \emph{blocked}). Then looking at the tiles from left to right,
   tiles immediately to the right of one that is blocked will either be blocked as well or will move
   left to merge with the blocked one if they can combine. A tile next
   to one that moves or next to an empty space will move one cell to
   the left. New tiles appear according to an unknown rule, which
   seems to change depending on the version of the game. It seems to
   be always of low value and on a cell next to a wall. At the end of
   the game, the score is computed by totalling $3^{i+1}$ points for each
   tile of value $3 \cdot 2^i$ on the board.
\item[Fives] was the first clone of Threes! for Android devices. Its
  rules are nearly identical to Threes!, except that the base tiles
  have values $2$ and $3$ which combine to form tile $5$, all other
  tiles have values $5 \cdot 2^i, i\geq 0$.
\item[1024!] The main difference with 2048 is that there is a
  fixed block in the middle of the board that doesn't move during the
  game. A new tile appears after each turn at a random location on the
  board (not necessarily next to a wall). The goal is to reach 1024.
\item[Saming's 2048] Just as in Cirulli's 2048, tiles are powers of
  two, however the tiles move according to slightly different
  rules. During a {\rm left} move for example, tiles are considered in
  each row from right to left. A tile $t$ will only move if its
  left neighboring cell is empty or of identical value. If there is no
  tile left of $t$, then $t$ is moved to a cell adjacent to the left
  wall. Otherwise let $s$ be the next tile left of $t$.
  If $t$ and $s$ are of identical values, the two
  cells are merged (and the value doubled), and the merged tile does
  not move this turn. Otherwise, the tile $t$ stops just to the right of
  $s$. A new tile of value 2 or 4 is inserted in a random empty cell
  at the end of the move.
\item[Det2048] This version is identical to 2048 except that 
a new tile always appears 
in the first empty cell (leftmost, then topmost) and its value is always \two. 
In its initial configuration, 
only a single tile \two is placed in the upper left cell
(Fig.~\ref{2048det}). 
\item[Fibonacci] In this popular version, the tiles have values from
  the Fibonacci sequence, and only tiles of successive values in the
  sequence are combined.
\end{description}
Other than those, the most natural variants use larger boards, and set higher
goal tile values.

The goal of this paper is to determine the computational complexity of
Threes!, 2048 and many of their variants. We follow the usual
offline deterministic model introduced by Demaine et al. for the
videogame Tetris~\cite{Tetris_IJCGA}. Given an initial configuration of tiles
in an $m\times n$ board, we assume the player has full knowledge of the
new pieces that will be added to the board after each
move\footnote{We ignore for now the issue of representing the position of
the new tiles, which might depend on which cells of the board are
empty, and thus on previous moves. As we will see, 
this will have little influence on the main results.}.
We prove that even with offline deterministic knowledge (and in all
variants listed above), it is \NP-hard to optimize several
natural objectives of the game:
\begin{itemize}
\item maximizing the largest tile created ({\sc Max-Tile}),
\item maximizing the total score ({\sc Max-Score}), and
\item maximizing the number of moves before losing the game ({\sc Max-Move}).
\end{itemize}

We show in fact that all three problems are inapproximable.
The decision problem for {\sc Max-Tile} is already \NP-hard for a
constant tile value, where the constant depends on the variant of the
game considered. On the other hand, {\sc Max-Move} is clearly
fixed-parameter tractable (\FPT)~\cite{downey1999parameterized}, that is, determining if
$k$ moves 
can be performed without losing takes only $O(4^k mn)$ time (since
there are only 4 moves possible at every step). 
Likewise, every merge increases the score by at least 4, and so the
number of moves to achieve score $x$ is at most $x/4$, therefore
determining if score $x$ can be achieved takes only $O(4^{x/4} mn)$
and {\sc Max-Score} is also \FPT.

\paragraph{\bf Related works.}
The tractability of computer games falls under the larger field of 
\emph{Algorithmic Combinatorial Game Theory} which has received
considerable interest over the past decade. See Demaine and
Hearn~\cite{AlgGameTheory_GONC3} for a recent survey. Block pushing and sliding
puzzles are probably the most similar to the games studied here, a
notable difference being that here (i) new tiles are (randomly) inserted after
each move and (ii) a merging mechanism reduces the number of tiles at
each step. Without these differences, the game would be nearly
identical to the \emph{Fifteen puzzle} and its generalizations, which
interestingly can be solved very efficiently.

%http://blog.openendings.net/2014/03/2048-is-in-np.html
%http://arxiv-web3.library.cornell.edu/abs/1408.6315
%http://arxiv.org/abs/1501.03837
In a blog post~\cite{2048inNP}, Christopher Chen proved that 2048 is in \NP, 
but for a variant where no new tile is inserted after each move.
More recently, two articles have appeared on arXiv with the aim of
analyzing the complexity of  Threes! and 2048. The first one~\cite{1408.6315} notes
2048 is \FPT (as discussed above) and claims {\sf PSPACE}-hardness of 
2048 by reduction from Nondeterministic Constraint Logic. 
However Abdelkader et al.~\cite{1501.03837} noted several issues with
that reduction. For instance, in order for their gadgets to function
properly, they need to modify the way tiles are moved and inserted
during the game. In particular, they allow tiles to be inserted by the
game at specific places in the middle of rows and columns in order to
maintain an invariant base pattern. Furthermore, the goal tile value
in the reduction is a (rather large) function of the input size. 
In an attempt to resolve these issues, Abdelkader et al.~\cite{1501.03837} analyze
the variant studied by Chen, in which no new tiles are generated during the game. For that
case, they show that for both Threes! and 2048, it is \NP-complete to
decide if a specific (constant) tile value can be reached. 

In the present paper, we analyze  Threes! and 2048, where new tiles
appear and tiles move and merge exactly as they do in the original
games. Our proofs are easily extended to most existing variants of the
game. We also prove inapproximability results for all these games, and
show most of them are in \NP.

\section{Definitions}

The simplest version of these games to describe is probably Det2048,
where after each move a tile \two appears in the first empty cell of
the board (in lexicographic order, leftmost, then topmost).

\begin{quote}
\begin{quote}
{\sc Make-\Xtile-Det2048} \\
Instance: 
An $m\times n$ board with an initial configuration of tiles, 
each of which has for value a power of 2, 
and a number $\X$, where $\X =2^c$ with some constant $c$. 
At the end of every turn, a new tile \two appears 
in the first empty cell in lexicographic order. \\
Question: 
Can one make \Xtile from the given configuration by a sequence moves 
(up, down, left and right)?
\end{quote}
\end{quote}

However in the original game, both tiles \two and \four can appear,
at a location determined by the game. Since we consider an offline model, the value
and location of the new tiles should be provided in the input. But
while encoding the value of the tile is easy,  its
location does not have such a natural representation, 
because the new tile can only be
inserted in an empty cell of the board, and the set of empty cells
depends on the previous moves in the game. 
One could conceive several
reasonable encodings (e.g., for each new tile, coordinates $(x,y)$
such that the new tile should be placed in the
closest/lexicagraphically first empty cell from cell $(x,y)$). 
To make our results as general as possible, we just assume
the \emph{location} information is encoded in constant space, and
the game uses that information to place the new tile.

\begin{quote}
\begin{quote}
{\sc Make-\Xtile} \\
Instance: 
An $m\times n$ board with an initial configuration of tiles, 
each of which has for value a power of 2, 
a number $\X$, where $\X =2^c$ with some constant $c$, and the sequence of
values (\two or \four) and location of the new tiles to be placed by
the game at the end of every turn. \\
Question: 
Can one make \Xtile from the given configuration by a sequence moves 
(up, down, left and right)?
\end{quote}
\end{quote}
In this setting, the original game 2048 is 
as {\sc Make-\Xtile} with $m=n=4$ and $\X=2048 = 2^{11}$ $(c=11)$. 

It will be useful to denote some of the variants of {\sc Make-\Xtile} 
by appending qualifiers to their name. For example, in the variant 
{\sc Make-\Xtile only-\two}, only tiles \two appear after each
move. In the {\sc Deterministic} variant, new tiles always appear in
the lexicographically first empty cell, and so 
{\sc Make-\Xtile only-\two Deterministic} is exactly {\sc Make-\Xtile-Det2048}.

For the purpose of analyzing the complexity of
the game, one might argue that the random nature of the original game
might make the game more (or less) tractable. 
To up the ante, some variants of the game, 
such as \emph{Evil2048}~\cite{evil2048}, 
use a heuristic to guess the worst possible location and value for the new
tile at the end of every move. On the other hand, for our hardness proofs it
might make sense to define a {\sc Make-\Xtile-Angel} version, where
the player can decide the value and location of the new tile
after every move. However, as we will see, none of this makes the game
significantly easier, as our \NP-hardness proofs and inapproximability
results hold for all variants mentioned, including 
{\sc Angel} version.

We will also define optimization problems {\sc Max-Tile}, 
{\sc Max-Score}, and {\sc Max-Moves}, whose objective is to maximize
the value of the maximum tile created in the game, the total score of
the game, defined as the sum of the values of all tiles created by
merges, and the number of moves played before losing the game
(reaching a forbidden configuration). These will be discussed in more
detail in the section on inapproximabililty.

Finally, variants using different merging and movement rules, such as
Threes! or Fibonacci will be defined and discussed after the main
\NP-hardness proof.

\section{\NP}\label{sec:np}
%http://blog.openendings.net/2014/03/2048-is-in-np.html
In a blog post~\cite{2048inNP} Christopher Chen showed that 2048 is in
\NP. However to simplify the proof, they assume no new piece gets added
to the board after a move. It turns out the proof for the regular
game ({\sc only-\two} version) is not much harder.

The complexity analysis of every problem depends on a reasonable 
representation of the input. We here assume the input is provided in
the form of $b$, the size of the board, and a list $L$ of tiles
present on the board at the beginning of the game. 
Thus the input size is $\log b + |L|$.

\begin{lemma}
For any constant value $\X$, {\sc Make-\Xtile only-\two} is in \NP
\end{lemma}
\begin{proof}
If the board is of size $b\times b$, then the maximum total value of
all the tiles on the board without ever reaching $\X$ is $\leq \X
b^2/2$. Since every move adds a tile with value \two, the total number of
moves without reaching $\X$ is $\leq \X b^2/4$. 
If the number of tiles in the starting configuration is $\geq b$, then
so is the input size, and the maximum number of moves reaching $\X$,
that is, the size of any yes certificate, is polynomial in the input size. 
If the number of tiles in the starting configuration is $< b$, then by
the pigeonhole principle there is an empty row. Therefore playing {\rm
  down} repeatedly will eventually, and repeatedly add new \two tiles
in that row which will accumulate to a single tile of value $2^b \geq \X$
for $b$ large enough.\qed
\end{proof}

For the more general version without the {\sc only-\two} restriction, 
it is plausible that a similar strategy would work.
An interesting question is whether the problem is still in \NP \ 
if $\X$ is
not a constant. The difficulty here is that the size of the input
could be as small as $\log \X$, and so number of moves would be
exponential in that. A good first step would be to settle the question
of what is the maximum tile value achievable on a $b\times b$ board in
the game Det2048. For now, the highest value, even on a $4\times 4$
board is unknown, the highest value was found using a heuristic
algorithm~\cite{det2048heuristic}. 
%http://codegolf.stackexchange.com/questions/24885/solve-a-deterministic-version-of-2048-using-the-fewest-bytes

\section{\NP-hardness}\label{sec:np-hardness}

We will show \NP-hardness of {\sc Make-\Xtile} 
by reduction from {\sc 3SAT}. 

\begin{quote}
\begin{quote}
{\sc 3SAT} \\
Instance: 
Set $U$ of variables, collection $C$ of clauses over $U$ 
such that each clause $c\in C$ has $|c|=3$. \\
Question: Is there a satisfying truth assignment for $C$?
\end{quote}
\end{quote}

\begin{theorem}
{\sc Make-\Xtile} is \NP-hard {\rm for any fixed $\X$ greater than 2048}. 
\end{theorem}

\proof
Reduction from {\sc 3SAT}. 
From an arbitrary instance of {\sc 3SAT} with $n$ variables and $m$ clauses, 
we construct a {\sc Make-\Xtile} instance.  
The construction will ensure that only tiles of value \two may be merged
(into \four tiles) 
except for two tiles of value $\X/2$ that can be combined to
obtain the target value \Xtile at the end of the game if and only if the {\sc 3SAT}
instance is satisfiable. 

In order to facilitate the analysis of the game, the instances
produced by the reduction will start will all cells of the board
filled with tiles, and maintain this invariant after each move 
(hereafter named \emph{fullness} invariant). 
In order to achieve this, we ensure that at most one pair of
\two tiles is adjacent on the board at all times (\emph{one-move} invariant) and no other pair of
identical tiles ever become adjacent during the game until the very last move. This forces the
player to make a binary choice: {\rm left/right} or {\rm up/down}.
In this manner, at the end of each move, only one cell next to a wall
is freed, and so the game (in every known variant) will have to place
a new tile in that exact cell.

We explain our construction by specifying the locations where tiles \two are placed. 
Since we always construct gadgets by putting tiles \two in pairs, 
we denote this by a pair of two 2D points $(\cdot ,\cdot)$. 
In the subsequent figures, 
they will be represented by black dots on a 2D lattice plane. 
The rest of the board will be filled with a pattern of tiles that will
prevent any accidental merges.

\paragraph{\bf Sketch.}
The construction has three parts: the \emph{variable} gadgets, the
\emph{literal} gadgets, and the \emph{clause checking} gadgets.
We place each variable gadget in a rectangle below the $x$
axis in distinct rows and columns (we use negative $y$ coordinates for
ease of notation). The clauses 
will each take up 12 rows above the $x$ axis, 4 for each literal. Each
literal will lie in 4 of the rows of its clause and 3 of the columns of its
variable. The variable gadgets will cause vertical ({\rm down}) shifts in their columns,
which will be transformed into horizontal ({\rm left}) shifts in the rows of
corresponding clauses by the literal gadgets.
Finally, the clause checking gadgets to the right of the board will
check, for each clause, that at least one of its rows has been shifted.

\paragraph{\bf Base Pattern.} 
We start the construction by filling the board with the
repeating \emph{base pattern} shown in Fig.~\ref{base_pattern}. 
Assuming the bottom left cell is numbered $(0,0)$,
cell $(i,j)$ of the base pattern contains the tile 
$2^{3(i\mod  3)+(j\mod 3)+3}$. See Fig.~\ref{base_pattern}.
\begin{figure}[hbt]
\centering
\scalebox{0.50}{\includegraphics{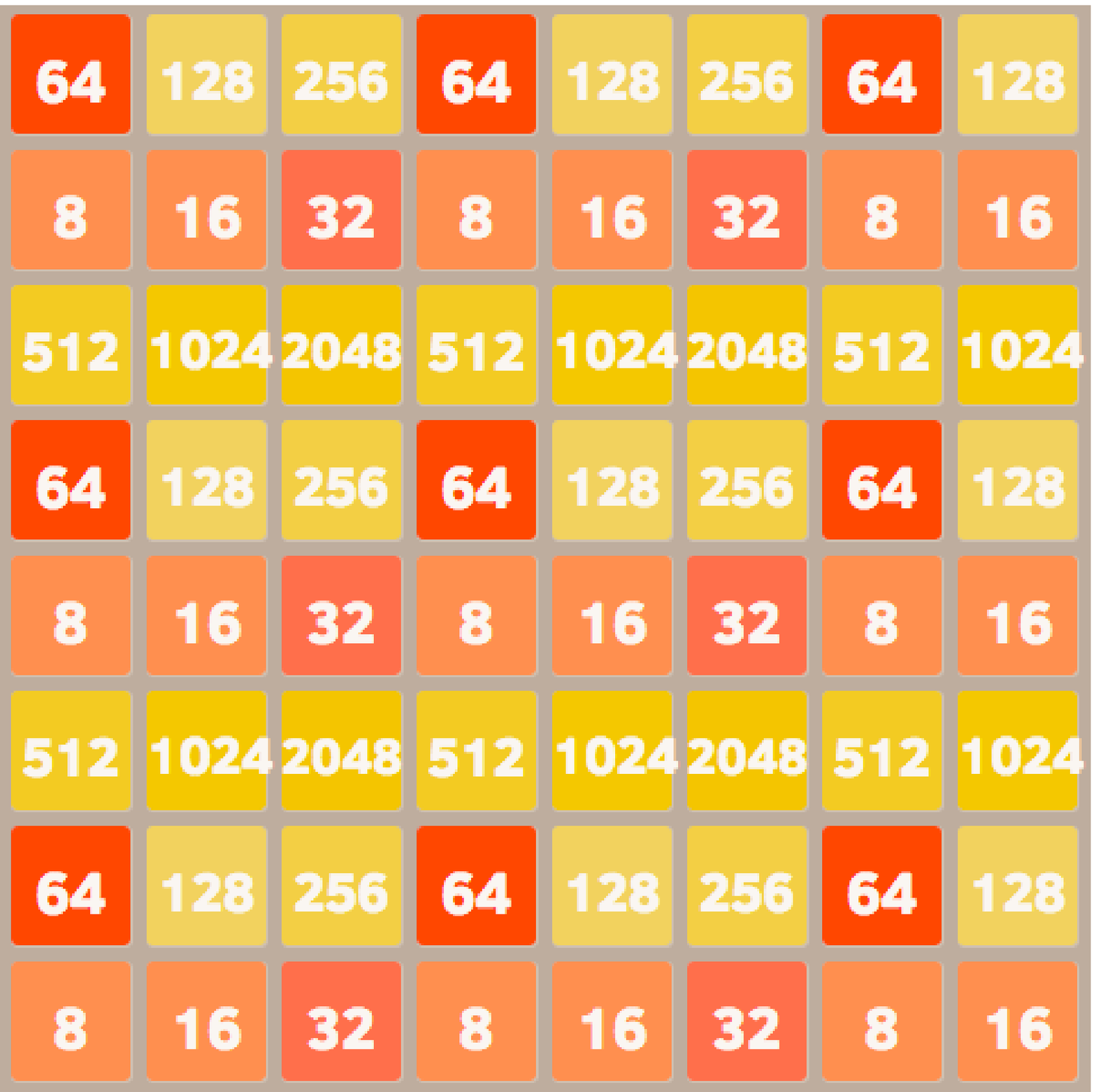}}
\caption{Base pattern.}
\label{base_pattern}
\end{figure}

The gadgets replace some of the cells by the tile \two.
This will cause some shifts in rows and columns during the game as
those tiles become adjacent.
But because of the one-move invariant, at most one row or column
shifts in each move.
In order to avoid accidental merges between tiles of the base pattern,
gadgets will be constructed in such a way that no row or column will
ever be shifted more than once (avoiding unwanted merges between new
tiles appearing in the same cell). 
Second, we will adjust the size of the board so that gadgets are at a
distance at least 3 from the walls, to avoid new tiles to interfere
with the \two tiles of the gadgets.
Furthermore, we will place gadgets in such a way that 
no two adjacent columns
will ever be shifted. Likewise, we will ensure that
no two adjacent rows will ever be shifted,
except in one place in the clause checking
gadget where the one-move invariant will have to be argued more
carefully. 

Ignoring this last case for now, notice that a tile can only be offset
by one position horizontally and one position vertically during the
entire game.
Two tiles of same value are adjacent if the difference between their
coordinates are $(0,1)$ or $(1,0)$, and any identical pair of tiles
that are not on the same row or the same column start with coordinate
difference at least $(3,3)$, and thus can never become adjacent.
If two identical tiles are on the same column, their vertical distance,
starting at $3$, would have to be reduced twice (using two opposite
vertical shifts) in order for them to be at distance $1$. 
However, since adjacent columns cannot be shifted in one game, this
can only happen if they are at horizontal distance $2$ at some point
during the game, but since they started at horizontal distance $0$,
making that happen would already spend the $2$ horizontal moves for
those two tiles, making it impossible to bring them back close enough
to become adjacent. The same argument applies to two identical tiles
on the same row at distance $3$.

Finally, consider the special case occurring in the clause checking
gadget, where two adjacent rows are shifted. A similar case analysis
will show that identical tiles on the same row could become adjacent,
but only if the last move is vertical. A careful inspection of the
gadget will reveal that no vertical move will occur within the shifted
portion of those rows after the adjacent horizontal shift occurs.

\paragraph{\bf Variables.}
For each variable gadget, we reserve 6 rows of the board, and 3
columns for each clause in which that variable appears. Assume, without
loss of generality, that every
variable appears at least once in the positive form (otherwise negate it).

Let $k_i^+$ and $k_i^-$ be the numbers of clauses 
containing $x_i$ and $\overline{x_i}$, respectively. 
We define the offset coordinates for each variable by 
\[
\left\{
\begin{array}{l}
X_0^{\rm V}=0, \\
X_{i}^{\rm V}=X_{i-1}^{\rm V}+3(k_i^+ + k_i^-)+7\ \ (1\le i\le n) 
\end{array}
\right. 
\mbox{ and}
\left.
\begin{array}{l}
Y_i^{\rm V}=-6i\ \ (0\le i\le n-1). 
%Y_i^{\rm V}=-6(i-1). 
\end{array}
\right. 
\]
Now for each variable $x_i$ $(i=1,\ldots ,n)$ 
we construct a gadget as follows. 
For a choice of true or false, we put a pair of tiles \two 
at $$(A^{\rm V}_i,B^{\rm V}_i) = \left( (X_{i-1}^{\rm V}+1, Y_{i-1}^{\rm V}+1), (X_{i-1}^{\rm V}+2, Y_{i-1}^{\rm V})\right).$$ 
For the false part we place two pairs of tiles \two at coordinates 
$$(C^{\rm V}_i,D^{\rm V}_i) = \left((X_{i-1}^{\rm V}+2, Y_{i-1}^{\rm V}-2), (X_{i-1}^{\rm V}+1, Y_{i-1}^{\rm V}-3)\right),$$ 
and
$$(E^{\rm V}_i,F^{\rm V}_i) = \left((X_{i-1}^{\rm V}+3k_i^+ +5, Y_{i-1}^{\rm V}-2), (X_{i-1}^{\rm V}+3k_i^+ +4, Y_{i-1}^{\rm V}-3)\right),$$
\begin{figure}[hbt]
\centering
\scalebox{1.15}{\input{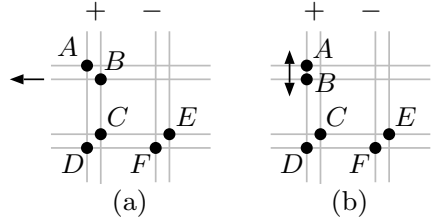}}
\caption{Variable gadget.}
\label{variable_gadget}
\end{figure}

So in the variable gadget for $x_i$, six tiles of value \two are placed 
as shown in Fig.~\ref{variable_gadget}(a) in general.
The gadget is \emph{activated} by pulling the row of $B_i$ left (i.e.,
by merging a pair of adjacent \two on the row of $B_i$ to the left).   
Now (see Fig.~\ref{variable_gadget}(b)), the tiles $A_i$ and $B_i$
become adjacent and on the same column.
At this point the player has the choice to move this column containing
$A_i$ {\rm down} (positive assignment to $x_i$ and to make each clause 
containing literal $x_i$ true), 
or {\rm up} (negative assignment to make each clause containing literal 
$\overline{x_i}$ true).
If the player chooses to move {\rm up}, then $D_i$ moves up one cell and
becomes adjacent to $C_i$, allowing the player to move {\rm left},
causing $E_i$ to move one cell left. 
The tile $E_i$ is now adjacent and on the same column as $F_i$ and the
column of $F_i$ can be moved down. Because there are no \two tiles
below the variable gadget or to its left, any sequence of moves other
than this one will cause the game to end.

Thus the $x_i$ variable gadget has for effect to move down either the column
of $A_i$ (true), or that of $F_i$ (false). This choice will be
propagated through each corresponding literal gadgets.

\paragraph{\bf Literals.} 
For literals in clauses we introduce the following coordinates: 
\[
\left\{
\begin{array}{l}
X_i^{\rm L}=X_i^{\rm V}\ \ (0\le i\le n), \\
X_{n+j}^{\rm L}=X_{n+j-1}^{\rm L}+25\ \ (1\le j\le m) 
\end{array}
\right. 
\mbox{ and}
\left.
\begin{array}{l}
Y_j^{\rm L}=12(j-1)+4\ \ (1\le j\le m). 
\end{array}
\right. 
\]
For variable $x_i$, 
suppose its positive literals $x_i$ appear 
in the $p_k$-th position of the $j_k$-th clause 
($p_k\in \{0,1,2\}$; $k=1,\ldots ,k_i^+$; $1\le j_1<\cdots <j_{k_i^+}\le m$). 
Remember that setting $x_i$ to true will shift column $X^V_{i-1}+1$
{\rm down}. The gadget for the first positive literal $x_i$ will
receiving this vertical activation, shift one of its rows left  
and propagate the {\rm down} move to activate the next literal.  
For this, we put two pairs of tiles \two at 
$$(A^{\rm L}_{j_k,p_k},B^{\rm L}_{j_k,p_k}) = \left((X_{i-1}^{\rm L}+3(k-1)+1, Y_{j_k}^{\rm L}+4p_k +1),
(X_{i-1}^{\rm L}+3(k-1)+2, Y_{j_k}^{\rm L}+4p_k)\right)$$ and 
$$(C^{\rm L}_{j_k,p_k},D^{\rm L}_{j_k,p_k}) = \left((X_{i-1}^{\rm L}+3(k-1)+4, Y_{j_k}^{\rm L}+4p_k -1), 
(X_{i-1}^{\rm L}+3(k-1)+5, Y_{j_k}^{\rm L}+4p_k)\right).$$ 

Likewise, when $k_i^- > 0$, negative literals $\overline{x_i}$ appear 
in the $p_k$-th position of the $j_k$-th clause 
($p_k\in \{0,1,2\}$; $k=1,\ldots, k_i^-$; $1\le j_1<\cdots <j_{k_i^-}\le m$). 
For receiving and propagating vertical activations 
we put two pairs of tiles \two at 
$$(A^{\rm L}_{j_k,p_k},B^{\rm L}_{j_k,p_k}) = \left((X_{i-1}^{\rm L}+3(k_i^+ +k)+1, Y_{j_k}^{\rm L}+4p_k +1), 
(X_{i-1}^{\rm L}+3(k_i^+ +k)+2, Y_{j_k}^{\rm L}+4p_k)\right)$$ and 
$$(C^{\rm L}_{j_k,p_k},D^{\rm L}_{j_k,p_k}) = \left((X_{i-1}^{\rm L}+3(k_i^+ +k)+4, Y_{j_k}^{\rm L}+4p_k -1), 
(X_{i-1}^{\rm L}+3(k_i^+ +k)+5, Y_{j_k}^{\rm L}+4p_k)\right).$$ 
See Fig.~\ref{literal_gadget}(a).
\begin{figure}[hbt]
\centering
\scalebox{1.15}{\input{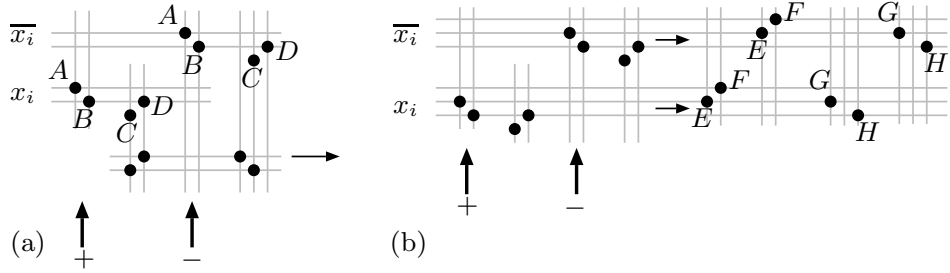}}
\caption{Literal gadget.}
\label{literal_gadget}
\end{figure}

The horizontal ({\rm right}) shifts for both positive and negative literals has for effect
to move \two tiles placed to the right of the board, for use in the
clause checking gadgets. We add those pairs of tiles \two at 
$$(E^{\rm L}_{j_k,p_k},F^{\rm L}_{j_k,p_k}) = \left((X_{n+j_k-1}^{\rm L}+6p_k+1, Y_{j_k}^{\rm L}+4p_k +1), 
(X_{n+j_k-1}^{\rm L}+6p_k+2, Y_{j_k}^{\rm L}+4p_k+2)\right)$$ and 
$$(G^{\rm L}_{j_k,p_k},H^{\rm L}_{j_k,p_k}) = \left((X_{n+j_k-1}^{\rm L}+3p_k+16, Y_{j_k}^{\rm L}+4p_k +1), 
(X_{n+j_k-1}^{\rm L}+3p_k+18, Y_{j_k}^{\rm L}+4p_k)\right).$$ 
See Fig.~\ref{literal_gadget}(b).

The final appearance of literals $x_i$ and $\overline{x_i}$ will also
be represented by two pairs of tiles like above, but the second pair
will cause a vertical shift {\rm up} which will be propagated to
activate the next variable gadget 
or to activate the clause checking gadgets 
as shown in Fig.~\ref{literal_gadget}(b). This process is described next.

\paragraph{\bf Activate.}
The first variable gadget is activated by a pair of \two placed
at 
$$\left((-3,0),(-2,0)\right)$$
causing a horizontal shift {\rm left} for $B^{\rm V}_1$. For subsequent
variables, assigning the truth value to the final literal $x_i$ or
$\overline{x_i}$ ($1\leq i\leq n-1$) will cause a vertical shift {\rm up} at column 
$X_{i-1}^{\rm L}+3(k_i^+-1)+4$ or 
$X_{i-1}^{\rm L}+3(k_i^++k_i^-)+4$.
Note that if $k_i^-=0$, then it is $E^V_i$ and $F^V_i$ which will
cause the vertical shift at that position.
We propagate that shift into a horizontal {\rm left} shift activating
variable $x_{i+1}$ using two pairs of tiles \two at 
$$\left((X_{i-1}^{\rm L}+3(k_i^+-1)+4, Y_{i-1}^{\rm V}-7), 
((X_{i-1}^{\rm L}+3(k_i^+-1)+5, Y_{i-1}^{\rm V}-6)\right)$$ and 
$$\left((X_{i-1}^{\rm L}+3(k_i^++k_i^-)+3, Y_{i-1}^{\rm V}-6), 
(X_{i-1}^{\rm L}+3(k_i^++k_i^-)+4, Y_{i-1}^{\rm V}-7)\right).$$
See the bottom four tiles of Fig.~\ref{literal_gadget}(a).

After the truth assignments of the literals of the last variable $x_n$ or
$\overline{x_n}$, one of the same columns is shifted, but this time
{\rm down} and that shift is propagated to activate the clause
checking gadgets using 
two pairs of tiles \two at 
$$\left((X_{n-1}^{\rm L}+3(k_n^+-1)+4, 12m+5), 
(X_{n-1}^{\rm L}+3(k_n^+-1)+5, 12m+4)\right)$$ and 
$$\left((X_{n-1}^{\rm L}+3(k_n^++k_n^-)+3, 12m+4), 
(X_{n-1}^{\rm L}+3(k_n^++k_n^-)+4, 12m+5)\right).$$

\paragraph{\bf Checking Clauses.} 
For clause checking gadgets
we take coordinates as follows: 
\[
\left.
\begin{array}{l}
X_{j}^{\rm T}=X_{n+j}^{\rm L}\ \ (0\le j\le m) 
\end{array}
\right. 
\mbox{and }
\left\{
\begin{array}{l}
Y_1^{\rm T}=12m+12, \\
Y_{j}^{\rm T}=Y_{j-1}^{\rm T}+15\ \ (1< j\le m). 
\end{array}
\right. 
\]
For each clause $C_j$ $(j=1,\ldots ,m)$ the corresponding gadget has for purpose to 
check that at least one literal of that clause has been set to
true.  
To choose which of the three literals will be checked, we  
we place the following five pairs of tiles \two at 
$$(A^{\rm T}_{j,1}, B^{\rm T}_{j,1}) = \left((X_{j-1}^{\rm T}+17, Y_j^{\rm T}-7), (X_{j-1}^{\rm T}+18, Y_j^{\rm T}-8)\right),$$ 
$$\left((X_{j-1}^{\rm T}+17, Y_j^{\rm T}+2), (X_{j-1}^{\rm T}+18, Y_j^{\rm T}+1)\right),$$ 
$$(A^{\rm T}_{j,2}, B^{\rm T}_{j,2}) = \left((X_{j-1}^{\rm T}+20, Y_j^{\rm T}+2), (X_{j-1}^{\rm T}+21, Y_j^{\rm T}+1)\right),$$ 
$$\left((X_{j-1}^{\rm T}+20, Y_j^{\rm T}+5), (X_{j-1}^{\rm T}+21, Y_j^{\rm T}+4)\right),$$
$$(A^{\rm T}_{j,3}, B^{\rm T}_{j,3}) = \left((X_{j-1}^{\rm T}+23, Y_j^{\rm T}+5), (X_{j-1}^{\rm T}+24, Y_j^{\rm T}+4)\right).$$ 
\begin{figure}[hbt]
\centering
\scalebox{1.15}{\input{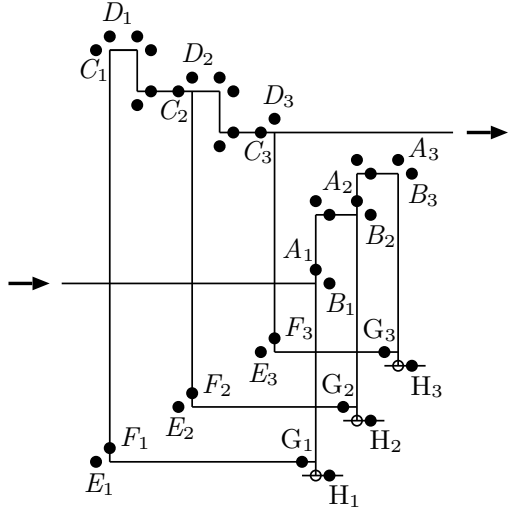}}
\caption{Clause checking gadget.}
\label{clause_test}
\end{figure}

The pair $(A^{\rm T}_{j,1},B^{\rm T}_{j,1})$ is activated by shifting the row $Y_j^{\rm T}-8$ 
{\rm left}. Then either the move {\rm up} shifts the column
$X_{j-1}^{\rm T}+17$ up, or the sequence {\rm down, left, up} shifts
the column $X_{j-1}^{\rm T}+20$ up, or the sequence
{\rm down, left, down, left, up} shifts column $X_{j-1}^{\rm T}+23$
up. Any other sequence of moves ends the game. 
Note that the column of $A^{\rm T}_{j,p}$, $p=1,2,$ or $3$ being
shifted up is exactly one column left of the one containing the \two at
$H^{\rm L}_{j,p}$ at the beginning of the game. If the corresponding literal
was set to true in the literal gadget, then the \two had been shifted
left and is now shifted up, bringing $G^{\rm L}_{j,p}$ and $H^{\rm L}_{j,p}$ next
to each other. The row of $G^{\rm L}_{j,p}$ can now be shifted {\rm left},
activating the pair $(E^{\rm L}_{j,p},F^{\rm L}_{j,p})$, and the column of
$F^{\rm L}_{j,p}$ can now be shifted {\rm down}.

To collect the {\rm down} shift in the column of $F^{\rm L}_{j,p}$ for the chosen
$p=1,2,$ or $3$, we place seven pairs of tiles \two at coordinates 
$$(C^{\rm T}_{j,1},D^{\rm T}_{j,1}) = \left((X_{j-1}^{\rm T}+1, Y_j^{\rm T}+13), (X_{j-1}^{\rm T}+2, Y_j^{\rm T}+14)\right),$$ 
$$\left((X_{j-1}^{\rm T}+4, Y_j^{\rm T}+14), (X_{j-1}^{\rm T}+5, Y_j^{\rm T}+13)\right),$$ 
$$\left((X_{j-1}^{\rm T}+4, Y_j^{\rm T}+9), (X_{j-1}^{\rm T}+5, Y_j^{\rm T}+10)\right),$$ 
$$(C^{\rm T}_{j,2},D^{\rm T}_{j,2}) = \left((X_{j-1}^{\rm T}+7, Y_j^{\rm T}+10), (X_{j-1}^{\rm T}+8, Y_j^{\rm T}+11)\right),$$ 
$$\left((X_{j-1}^{\rm T}+10, Y_j^{\rm T}+11), (X_{j-1}^{\rm T}+11, Y_j^{\rm T}+10)\right),$$ 
$$\left((X_{j-1}^{\rm T}+10, Y_j^{\rm T}+6), (X_{j-1}^{\rm T}+11, Y_j^{\rm T}+7)\right),$$ 
$$(C^{\rm T}_{j,3},D^{\rm T}_{j,3}) = \left((X_{j-1}^{\rm T}+13, Y_j^{\rm T}+7), (X_{j-1}^{\rm T}+14, Y_j^{\rm T}+8)\right).$$ 

Now the vertical shift of column $F^{\rm L}_{j,p}$ aligns the \two tiles of
$C^{\rm T}_{j,p}$ and $D^{\rm T}_{j,p}$, and the rest of the \two can be used to
propagate the horizontal shift until the row of $C^{\rm T}_{j,3}$ is shifted
left. This is the same row as $B^{\rm T}_{j+1,1}$ and so activates the next
clause checking gadget for $j<m$.

\paragraph{\bf Goal.} To make a target number $X$ ($> 2048$), 
we place a pair of tiles of value $X/2$ at 
$\left((X_{n+m}^{\rm T}+1, Y_m^{\rm T}+8), (X_{n+m}^{\rm T}+2, Y_m^{\rm T}+7)\right)$. 
This pair will become adjacent when the last clause checking gadget is
successfully played and shifts the row of $C^{\rm T}_{m,3}$ {\rm left}.

From the construction, it is clear that the tile \Xtile can be created
if and only if given {\sc 3SAT} formula 
is satisfiable. 
The size of the board is $\Theta((n+m)^2)$, and the size of the
sequence of new tiles is $\Theta(m+n)$.
So this reduction takes polynomial space and polynomial time 
with respect to the input size $n+m$ of the {\sc 3SAT} instance. 
\qed

\bigskip
We illustrate a complete example of our reduction  
in Fig.~\ref{reduction}, where the formula for {\sc 3SAT} is 
$f=(x_1\vee \overline{x_2}\vee x_3)\wedge (x_2\vee x_3\vee \overline{x_4})
\wedge (\overline{x_1}\vee \overline{x_2}\vee x_4)$. 
In this figure, two goal tiles $\X/2$ are represented by squares.

\begin{figure}[hbt]
\centering
\scalebox{0.77}{\input{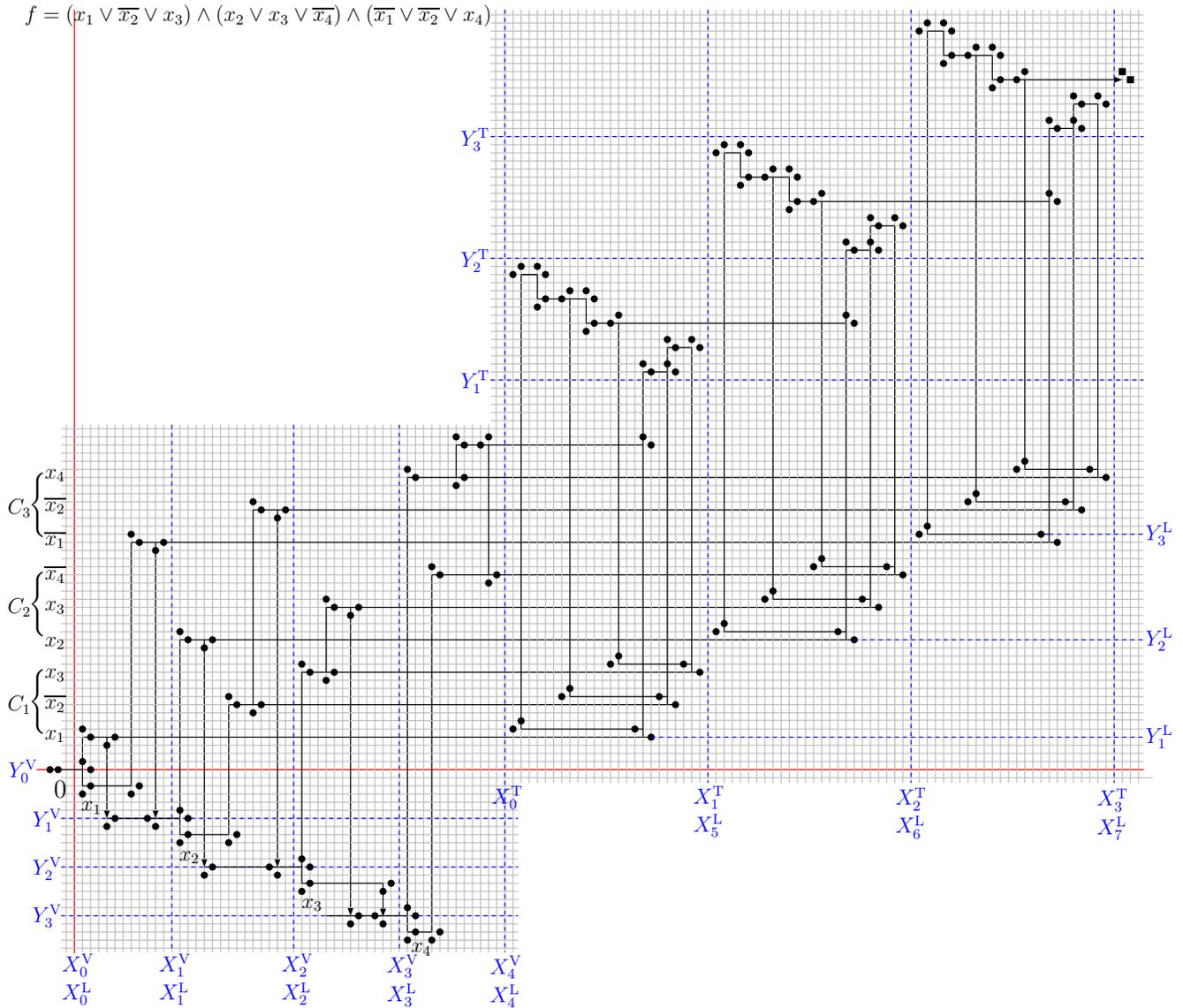}}
\caption{An example of \NP-hardness reduction from {\sc 3SAT} to {\sc Make-\Xtile}.}
\label{reduction}
\end{figure}

\section{Inapproximability}\label{sec:approximability}

It is fairly easy to extend the construction from the previous section
to show it is \NP-hard to approximate {\sc Max-Tile}, {\sc Max-Score} or
{\sc Max-Moves}. For {\sc Max-Tile} and {\sc Max-Score}, it would be
enough to change the value of the goal pair of tiles to an arbitrarily
high number. However one might want to impose that the tiles of the
input configuration be all of small value. In that case, we can still
show inapproximability by using the \emph{pot of gold} technique. 

Note that in the previous construction, if the formula is satisfiable,
then the goal tiles will be adjacent in column $X^{\rm T}_{n+m}+1$. We add
tiles:

$$(A^{\rm A},B^{\rm A}) = \left((X^{\rm T}_{n+m}+1, Y^{\rm T}_m+21),(X^{\rm T}_{n+m}+2, Y^{\rm T}_m+20)\right)$$

\noindent
We then extend the board to the left of the first variable gadget by
$K=2^p$ columns, and place tiles of alternating values \tileVIII and
\tileXVI on row $Y^{\rm T}_m+20$ at negative $x$ coordinates. On the row
$Y^{\rm T}_m+19$, we place tiles of alternating values \tileXXXII and
\tileVIII (so the tiles \tileVIII are just below the tiles \tileXVI
and can't merge. 

If the formula is satisfiable (and only then), the player can solve
the game as before until she shifts column $X^{\rm T}_{n+m}+1$ {\rm down}. 
One can then shift the row of $B^{\rm A}$ {\rm left} which aligns all the
\tileVIII of that row with the \tileVIII of the row below. A move {\rm
  up} now merges all those {\tileVIII}s into {\tileXVI}s, and
repeatedly shifting right $p=\log K$ times will merge all those tiles
into one tile of value $16K$. We can then continue the sequence with
$S=2^q$ \two appearing in the leftmost cell of row $Y^{\rm T}_m+20$, with $q<K$.
The total score is then $\Theta(m+n+K+S)$ and the maximum tile is
$\Theta(\max(K,S))$.

The input size in this game is
the entire size of the board plus the length of the tile
sequence and all tiles are of constant value. 
The original board size is $\Theta((n+m)^2)$ so the augmented board
is of size $\Theta((n+m)(n+m+K))$. The number
of moves is $\Theta(n+m+\log K+S)$.
 
So in the standard game:

\begin{itemize}
\item Taking $K=n+m$, $S=K^3$, the input
  size is $N=K^3+\Theta(K^2)$. The maximum tile value is $S=N-O(N^{2/3})$ if the formula is satisfiable, 2048
otherwise. So, it is \NP-hard to approximate {\sc Max-Tile} within a factor
$N/c$ for some constant $c$.
\item Taking $K=n+m$, $S=2^K$, the input
  size is $N=2^K+\Theta(K^2)$. The maximum score is $S=N-O(\log^{2}N)$
  if the formula is satisfiable, $O(K)=O(\log N)$ otherwise. So it is
  \NP-hard to approximate {\sc Max-Score} within a factor 
  $o(N/\log N)$.
\item Using the same parameters as above, the maximum number of moves
  is at least $S=N-O(\log^{2}N)$ if the formula is satisfiable,
  $O(K)=O(\log N)$ otherwise. So it is
  \NP-hard to approximate {\sc Max-Moves} within a factor 
  $o(N/\log N)$.
\end{itemize}

Note the importance of $S$ in the input size $N$ for the standard version
of the game. In the game {\sc Det2048}, however, the sequence of new
tiles is implicit and not part of the input. The inapproximability
results are then strengthened: all three problems are inapproximable
within a factor $o(2^N)$ or $o(2^N/N)$.

\section{Variants}\label{sec:variants}
Only minor modifications are required to make the \NP-hardness
reduction work for most known variants of Threes! and 2048. We just
describe them for the original game Threes!, and for the Fibonacci
version mentioned in the introduction. The extension of these results
to other variants such as Fives, 1024! and Saming's 2048 are immediate
and are left as an exercise to the reader.

\subsection{Threes!}
The reduction for Threes! is nearly identical as for 2048. The easiest
way to repeat the proof would be to replace every tile of value $2^a$
by a tile of value $3 \cdot 2^{a-1}$. A slightly better bound can be
obtained in the following manner. 
Every occurence of tile \two is replaced by a \threethrees. The base pattern uses
only tiles \onethrees, and the new tiles added after each move are
\onethrees. Since a \onethrees can only merge with a \twothrees, this will ensure the
base pattern never causes an unwanted merge. The goal tiles are
replaced by two \sixthrees.
 
Recall that tiles in Threes! move according to slightly different
rules (most importantly, every tile stays in place, shifts to an
adjacent cell or merges with an adjacent tile in every move). 
However, because of the fullness and one-move invariants, the result
of a move will be the same in Threes! as it was for 2048.

Therefore, an identical proof shows that it is \NP-hard to decide if it is
possible to achieve tile \twelvethrees in Threes!. 
The inapproximability results for {\sc Max-Tile} and {\sc Max-Moves}
extend as well. For {\sc Max-Score}, the situation is even worse, as
following the same reduction ending it with a sequence of $S=2^q$
tiles \threethrees, we would produce a tile of value $3\cdot 2^q$
which will produce a score of $3^{q+1} = 3S^{\log_2 3} =
\Omega(N^{\log_2 3}$ if the formula
is satisfiable, and $O(K)=O(\log N)$ otherwise. Therefore, it
is \NP-hard to approximate {\sc Max-Score} in Threes! within a factor 
  $o(N^{\log_2 3}/\log N)$.

\subsection{Fibonacci}
Denote the $i$-th Fibonacci number by $F_i$, 
that is, $F_1=F_2=1$, and $F_{i+2}=F_{i+1}+F_{i}$. 
In the Fibonacci version, tiles merge only if they are adjacent in the
Fibonacci sequence.

We modify the reduction so that every occurence of tile \two is
replaced by a 1 (since $F_1=F_2=1$, they can merge into a 2 when
adjacent). 
The base pattern uses only tiles of value 5, and the new tiles added after each move are
1. Since a 5 can only merge with a 3 or 8, this will ensure the
base pattern never causes an unwanted merge. The goal tiles are
replaced by 13 and 21. Therefore is \NP-hard to decide if it is
possible to achieve tile 34. Inapproximability results extend as well.

%\section*{Acknowledgments}

%Yushi Uno was supported in part by KAKENHI Grant numbers 23500022 and 25106508. 
\bibliographystyle{plain}
\bibliography{2048}

\end{document}